\documentclass[a4paper,twocolumn,11pt]{quantumarticle}
\pdfoutput=1
\usepackage[utf8]{inputenc}
\usepackage[english]{babel}
\usepackage[T1]{fontenc}
\usepackage{amsmath}
\usepackage{amssymb}
\usepackage{amsthm}
\usepackage{amsfonts}
\usepackage{cancel}
\usepackage{caption}
\usepackage{longtable}
\usepackage{subcaption}
\usepackage{enumitem}
\usepackage{multirow}
\usepackage{tcolorbox}
\usepackage{comment}
\usepackage{braket}
\usepackage{graphicx}
\usepackage{algpseudocode}
\usepackage{booktabs}
\usepackage{tikz}
\usetikzlibrary{quantikz2}
\usepackage{hyperref}
\usepackage{lineno}
\usepackage{orcidlink}
\usepackage[numbers]{natbib}

\newtheorem{definition}{Definition}
\newtheorem{theorem}{Theorem}

\newtheorem{corollary}{Corollary}
\newtheorem{lemma}[theorem]{Lemma}
\newtheorem{example}{Example}

\begin{document}

\title{Extending the Bloch sphere model to an N-qubit system}

\author{Francisco Piñero}
\orcid{0009-0003-6978-910X}
\email{Francisco.Pinero@uclm.es}
\thanks{Corresponding author}
\affiliation{Universidad de Castilla-La Mancha, Albacete, Spain}

\author{Cristian Franco}
\orcid{0009-0001-4842-7834}
\email{cristian.franco@unir.net}
\affiliation{Universidad de La Rioja, Logroño, Spain}

\author{Hernán I. de la Cruz}
\orcid{0000-0001-6445-5256}
\email{hernanindibil.lacruz@alu.uclm.es}
\affiliation{Universidad Complutense de Madrid, Madrid, Spain}

\author{Fernando L. Pelayo}
\orcid{0000-0001-7849-087X}
\email{fernandol.pelayo@uclm.es}
\affiliation{Universidad de Castilla-La Mancha, Albacete, Spain}

\author{Vicente Pascual}
\email{vpfuniversity@gmail.com}
\affiliation{Universidad de Castilla-La Mancha, Albacete, Spain}

\author{Mauro Mezzini}
\orcid{0000-0002-5308-0097}
\email{mauro.mezzini@uniroma3.it}
\affiliation{Roma Tre University, Roma, Italy}

\author{Jose Javier Paulet}
\orcid{0000-0002-0777-4119}
\email{paulet@qsimov.com}
\affiliation{QSimov, Talavera de la Reina, Spain}

\author{Fernando Cuartero}
\orcid{0000-0001-6285-8860}
\email{fernando.cuartero@uclm.es}
\affiliation{Universidad de Castilla-La Mancha, Albacete, Spain}

\maketitle

\begin{abstract}
    \noindent
The Bloch sphere is an elegant tool for representing single-qubit states. However, a widely accepted generalization for multi-qubit systems with entanglement remains absent. We propose a novel geometric model extending the Bloch sphere representation to arbitrary $N$-qubit systems using $2^N-1$ spheres. We demonstrate that any pure 2-qubit state is uniquely described by three spheres: two for individual qubits and a third encapsulating bipartite entanglement. Generalizing this, we establish an $N$-qubit parameterization through the hierarchical application of controlled rotation gates along the $Z$ and $Y$ axes. We formally prove a strict bijection between the standard state vector representation and our model's angular parameters. This framework provides an intuitive visualization of multiple entanglement, offering potential computational advantages for quantum simulators and new analytical perspectives on quantum gates.
\end{abstract}

\section{Introduction}

The idea of a computer based on quantum theory was first formulated by Richard Feynman \cite{Feynman} in 1981, during his talk ``Simulating physics with computers''. Feynman observed that simulating the evolution of quantum physical systems (elementary particles, atoms, or molecules) was computationally very expensive for classical machines. To solve this, he proposed harnessing the power of the, by then no longer so young, quantum theory itself to build a new type of computer capable of simulating nature natively.

To rigorously understand the foundations of this computation, we must be clear about certain concepts of quantum mechanics. This is a physical theory that allows us to make predictions about certain types of experiments involving an interaction between physical systems and certain measuring devices acting upon them \cite{isham1995lectures}. Furthermore, the mechanics is structured in two layers: the deepest and most fundamental layer is an abstract probability theory, Quantum Probability Theory (QPT) \cite{Pl_vala_2023}, and upon it rests quantum physics itself. QPT establishes the mathematical rules dictating how states, transformations, observables, and systems are represented, as well as how to compute the probabilities of the possible outcomes after performing measurements. It is precisely this abstract probabilistic foundation that we need to understand quantum information and computation.

Following the rules of QPT, the first postulate of quantum mechanics states that the state of an isolated system is mathematically represented by a vector in a complex Hilbert space \cite{sakurai2017modern}.

The quantum system that interests us the most is the two-level system: the qubit. We can express an arbitrary pure state of a single qubit $\ket{\psi}$ as a linear combination of the orthonormal basis vectors $\ket{0}$ and $\ket{1}$:

\begin{equation}
\label{eq1}
    \ket{\psi} = a \ket{0} + b\ket{1}  \hspace{2em} a, b \in \mathbb{C}, ~ |a|^2 + |b|^2 = 1,
\end{equation}
where $a$ and $b$ are the complex probability amplitudes associated with measuring the system in states $\ket{0}$ and $\ket{1}$, respectively. Being probability amplitudes, it is imperative that the sum of all partial probabilities is 1. Before continuing, it is crucial to emphasize that two vectors differing only by a global phase (i.e., $\ket{\psi}$ and $e^{i\varphi}\ket{\psi}$) represent the same quantum state, as they yield exactly the same probability distribution when measured and are therefore physically indistinguishable.

The conservation of probabilities and the absence of physical meaning for a global phase allow us to express our qubit in a much more graphical way: by means of the widely known Bloch sphere (see Figure \ref{fig:BlochS}).
\small{
\begin{equation}
\label{eq2}
    \ket{\psi} = \cos \frac{\theta}{2} \ket{0} + e^{\varphi  i } \sin \frac{\theta}{2} \ket{1}\hspace{1em} \theta \in [0,\pi], \varphi \in [0, 2\pi).
\end{equation}
}

\begin{figure}[htbp]

    \centering

    \includegraphics[width=1\linewidth]{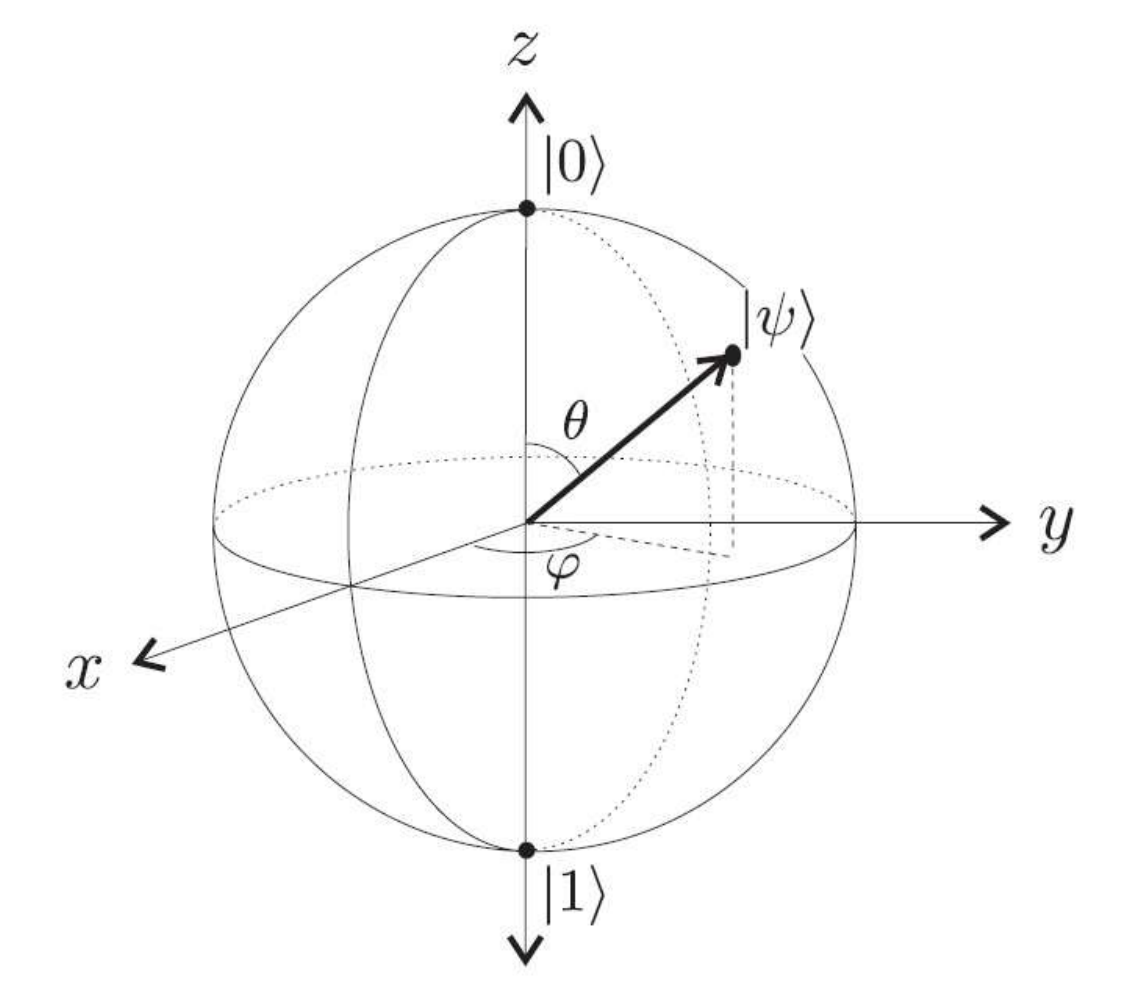}

    \caption{\label{fig:BlochS} State of a single qubit represented on the Bloch sphere.}

\end{figure}

With this parameterization, we can understand any pure state of a qubit as points on the surface of the Bloch sphere. The points inside the sphere itself are the so-called mixed states, but we will not delve into them here.

Within this theoretical framework, unitary operations, such as quantum logic gates or the time evolution of Hamiltonians, translate into rotations around the axes of the sphere. The fact that this representation is so visual and clear has made the Bloch sphere an indispensable tool for understanding the dynamics of single-qubit systems.

However, as we scale the system by increasing the number of qubits, this visual representation disappears. When we have systems with more than one qubit, entanglement appears, and we are no longer able to express the state of these systems using only the Bloch spheres of each individual qubit. To date, there is no widely accepted representation as intuitive and visual as the Bloch sphere for entangled multi-qubit systems, despite several attempts in the literature. 

Among these attempts, some preliminary work\cite{Mosseri} generalizes the standard Bloch sphere representation for two qubits within the framework of Hopf fibrations, requiring a 7-dimensional sphere. In \cite{Havel} it is shown that the states of a system of two qubits can be represented in a 6-dimensional geometric algebra quite similar to the Bloch Sphere. In \cite{wharton} it is reported that any pure two-qubit state can be represented by six real angles, with a natural parameterization induced by the bipartite structure. Up to a certain point this is a result close to ours, but there are some differences that we remark on in the following. And finally, we can cite \cite{Wie}, where a model of three Bloch Spheres is presented as a model for a 2 qubits system, by means of Hopf fibrations, as the first work cited, besides the fact that one of the spheres is not a unit sphere, thus they need an additional dimension to represent the radius of the sphere, so getting in the end 7 dimensions. Wang \cite{wang} applied a geometric algebra to analyze the space of a multi-qubit system, particularly two and three, so using a single angle to represent the entanglement in terms of the Von Newman entropy. Finally, in \cite{Dilley22} a model of a two qubit system is presented, where two spheres and a $3 \times 3$ correlation matrix are used.

\subsection{Contributions}

To bridge this gap, this paper proposes a novel and systematic extension of the Bloch sphere model capable of representing general $N$-qubit systems. We demonstrate that the complex amplitudes of an entangled composite state can be fully captured by introducing a hierarchy of ``entanglement spheres''. Specifically, while an isolated single qubit requires one sphere, a general 2-qubit system can be completely modeled using three spheres: one for each individual qubit, and a third sphere dedicated strictly to capturing their bipartite entanglement. 

Generalizing this principle, we show that any arbitrary $N$-qubit system can be completely and bijectively described using exactly $2^N-1$ spheres. Despite the inherent mathematical complexity of multi-qubit Hilbert spaces, this parameterization offers a surprisingly straightforward and highly structured geometric alternative, providing a clear advantage for describing, analyzing, and conceptualizing quantum states.

\section{Results}

\subsection{The model for two separated qubits}

A general state that represents a system of two separated qubits is
\small{
\begin{equation}
\label{eq5}
\ket{\psi}=\ket{Q_1}\otimes\ket{Q_2} = \left(a\ket{0}+b\ket{1}\right)\otimes\left(c\ket{0}+d\ket{1}\right),
\end{equation} 
}
where $a, b, c, d \in \mathbb{C}, ~ |a|^2 + |b|^2=1 $ ~$ |c|^2 + |d|^2 = 1$.

As we have seen, we can express each qubit with the Bloch sphere parameterization, so the two qubit separable state will be:
\begin{align} 
\label{eq6}
&\ket{Q_1} = \cos \frac{\theta_1}{2} \ket{0} + e^{ i \varphi_1 } \sin \frac{\theta_1}{2} \ket{1} \nonumber \\
&\ket{Q_2} = \cos \frac{\theta_2}{2} \ket{0} + e^{ i \varphi_2 } \sin \frac{\theta_2}{2} \ket{1} \nonumber \\
&\ket{Q_1} \otimes \ket{Q_2} = \cos \frac{\theta_1}{2} \cos \frac{\theta_2}{2} \ket{00} +  \nonumber \\
&\quad +\cos \frac{\theta_1}{2} \sin \frac{\theta_2}{2}  e^{i \varphi_2} \ket{01}+ \sin \frac{\theta_1}{2} \cos \frac{\theta_2}{2}  e^{i \varphi_1} \ket{10}+ \nonumber \\
&\quad + \sin \frac{\theta_1}{2} \sin \frac{\theta_2}{2}  e^{i (\varphi_1 + \varphi_2)} \ket{11} .
\end{align}
In this case, as we do not have any kind of entanglement, only 4 real dimensions are required to represent the system.

\subsection{Modelling the general 2-qubit system}

As discussed in the introduction, existing models for a two-qubit system require 6 or 7 real dimensions. An arbitrary two-qubit pure state is defined by 4 complex amplitudes, equating to 8 real parameters. However, by factoring in the normalization condition and disregarding the physically meaningless global phase, the system is reduced to 6 independent real parameters. Now we demonstrate that these 6 dimensions can be intuitively captured using three Bloch spheres. Two of these spheres, defined by the angles $(\theta_1, \varphi_1)$ and $(\theta_2, \varphi_2)$, encode the state of each individual qubit, while the third sphere, defined by $(\theta_3, \varphi_3)$ , is dedicated to modeling the bipartite entanglement. Finally, we will demonstrate that we can encode our sphere parameterization in a quantum circuit using controlled rotation matrices, which will allow us to lay the groundwork for the generalization to systems with multiple qubits.

\begin{theorem}\label{th1}
     Let $\ket{Q_1 Q_2}$ be a system of two qubits, represented by a vector in the computational basis with complex components $[x_1,x_2,x_3,x_4]$. Then there exist three Bloch spheres of coordinates $(\theta_j,\varphi_j)$, $j=1,2,3$, where $\theta_j \in [0, \pi]$ and $\varphi_j \in [0, 2\pi)$, such that 

\begin{align*}
    &x_1=\cos \frac{\theta_1}{2}\cos\frac{\theta_2}{2}\\
    &x_2=\cos \frac{\theta_1}{2}\sin\frac{\theta_2}{2}e^{i\varphi_2}\\
    &x_3=\sin \frac{\theta_1}{2}\cos\left(\frac{\theta_2}{2}+\theta_3\right)e^{i\varphi_1}\\
    &x_4=\sin \frac{\theta_1}{2}\sin\left(\frac{\theta_2}{2}+\theta_3\right)e^{i(\varphi_1+\varphi_2+\varphi_3)}.\\
\end{align*}
\end{theorem}

\begin{proof}
We can write the elements of the computational basis of $\ket{Q_1Q_2}$ as $x_1=r_1$, $x_2=r_2e^{i\psi_1}$, $x_3=r_3e^{i\psi_2}$ and $x_4=r_4e^{i\psi_3}$, where we have disregarded the physically meaningless global phase by taking $r_j$, $j=1,\dots,4$ as positive real parameters. For the phases, we have that the system
\begin{equation}
    \left\{
      \begin{aligned}
        &\psi_1=\varphi_2 \\
        &\psi_2=\varphi_1\\
        &\psi_3=\varphi_1+\varphi_2+\varphi_3
      \end{aligned}
    \right. 
\end{equation}
can always be solved by taking $\varphi_1=\psi_2$, $\varphi_2=\psi_1$, and $\varphi_3=\psi_3-\psi_2-\psi_1$. Furthermore, we can add a phase shift of $\pi$ to these phases in case any of the moduli do not have the correct sign (we will look into this in depth shortly).

Let us now consider the moduli. We must prove that the following system of equations has a solution for any $r_j$, $j=1,\dots,4$:

\begin{equation}
    \left\{
      \begin{aligned}
        &r_1=\cos \frac{\theta_1}{2}\cos\frac{\theta_2}{2} \\
        &r_2=\cos \frac{\theta_1}{2}\sin\frac{\theta_2}{2}\\
        &r_3=\sin \frac{\theta_1}{2}\cos\left(\frac{\theta_2}{2}+\theta_3\right)\\
        &r_4=\sin \frac{\theta_1}{2}\sin\left(\frac{\theta_2}{2}+\theta_3\right).
      \end{aligned}
    \right. \label{eq:proof1_moduli1}
\end{equation}
Additionally, the normalization condition must be met:
\begin{equation}
    \sum_{j=1}^4r_j^2=1.
    \label{eq:proof1_normalization}
\end{equation}

It is easy to verify that, if we substitute the first three equations of \eqref{eq:proof1_moduli1} into \eqref{eq:proof1_normalization}, we obtain the fourth equation of \eqref{eq:proof1_moduli1}, so we can disregard it and focus only on the first three. By operating, we arrive at the following solutions: 

\begin{equation}
    \left\{
      \begin{aligned}
        &\cos^2(\frac{\theta_1}{2})=r_1^2+r^2_2 \\
        &\cos^2(\frac{\theta_2}{2})=\frac{r_1^2}{r_1^2+r_2^2}\\
        &\cos^2\left(\frac{\theta_2}{2}+\theta_3\right)=\frac{r_3^2}{1-(r_1^2+r_2^2)},\\
      \end{aligned}
    \right. \label{eq:proof1_moduli2}
\end{equation}

which are valid as long as $r_1^2+r_2^2\neq0$ and $r_1^2+r_2^2\neq1$. When $r_1^2+r_2^2=0$, it means that $r_1=r_2=0$ and, by \eqref{eq:proof1_moduli1}, it means that $\theta_1=\pi$ and $\theta_2$ can take any value. To ensure our parameterization is unique, in these cases we will take $\theta_2=0$ as the canonical value. Similarly, if $r_1^2+r_2^2=1$, by \eqref{eq:proof1_normalization} we have that $r_3^2=r_4^2=0$ and $\theta_1=0$, meaning $\theta_3$ is arbitrary. Here we will also take $\theta_3=0$ as the canonical value.

We must realize that, although $r_3$ is non-negative by definition, since $\left(\theta_2/2+\theta_3\right)\in [0, 3\pi/2]$, the term $\cos\left(\theta_2/2+\theta_3\right)$ in \eqref{eq:proof1_moduli1} can take negative values. However, this does not jeopardize our proof because, as mentioned earlier, any unwanted sign when constructing our state vector from the parameterization angles can be offset with a $\pi$ phase shift in the corresponding phase.

\end{proof}

\bigskip
Having established our parameterization in Theorem \ref{th1}, we now demonstrate how it is easily translated into a quantum circuit. To achieve this, we take as our starting point the separable state (Eq. \ref{eq6}) and systematically introduce the entanglement by manipulating both the phase and the modulus of the amplitudes. First, however, we must properly define our working tools.

In standard quantum circuit design, as shown in \cite{Nielsen}, an arbitrary single-qubit state can be reached via a sequence of rotations along the $Z$ and $Y$ axes:

\begin{align} 
\label{eq7}
R_z(\varphi) &= \begin{pmatrix}
1 & 0 \\
0 & e^{i \varphi}
\end{pmatrix} & 
R_y(\theta) &= \begin{pmatrix}
\cos \theta & - \sin \theta \\
\sin \theta & \cos \theta 
\end{pmatrix}.
\end{align}

Our definition of rotation in $Z$ differs from Nielsen and Chuang's in a global phase in order to keep the first term of the state vector phase-free.

To extend these rotations to multiqubit systems and capture entanglement, we adopt a notation for controlled quantum gates inspired by Barenco et al \cite{Barenco95}.

\begin{definition}
     Let $\ket{Q_1 \ldots Q_n}$ be a system of $n$ qubits, let $U$ be a unitary gate acting on a single qubit, let $i \in \{2 \ldots n\}$ be the index inside the register of the qubit on which $U$ acts, and let $J \subset \{1 \ldots i - 1\}$ be a set of qubits. We denote the gate acting on qubit $i$ controlled by qubits from the set $J$ as $\bigwedge_i^J U$. 
\end{definition}

If $J= \emptyset$, we have a not controlled gate, thus $\bigwedge_i^\emptyset U = I^{\otimes^{i-1}} \otimes U \otimes I^{\otimes^{n-i-1}}$. With these mathematical tools defined, we can proceed to build the composite state.

\vspace{0.5cm}

{\bf Entanglement in the phase component}

Let us first analyze the phase components of the basis vectors in a separable state. The state $\ket{00}$ acts as our phase reference and lacks a relative phase. The amplitudes of $\ket{01}$ and $\ket{10}$ inherit the phases of qubits $Q_1$ and $Q_2$ respectively. Consequently, in a separable state, the phase of the $\ket{11}$ amplitude is constrained, being the sum of the two preceding phases.

However, for an arbitrary, non-separable state vector $(a, b, c, d) \in \mathbb{C}^4$, the phase of the fourth component can be independent. To capture this independence, we must introduce an additional parameter, the ``entanglement phase''. 

In order to establish a notation clear enough to allow for generalization to multi-qubit systems, we will denote the angles of the controlled rotations (i.e., of the entanglement spheres) as $\{\varphi_{J,k},\theta_{J,k}\}$, where $J$ is the set of control qubits and $k$ is the target qubit. Thus, our bipartite entanglement phase will be denoted as $\varphi_{1,2}$.

Using our previously defined notation, the introduction of this entanglement phase, which only acts upon the $\ket{11}$ amplitude, is represented by a $Z$-rotation on qubit 2 controlled by qubit 1. The matrix representation for $\bigwedge_2^1 R_z(\varphi_{1,2})$ is defined as:

\begin{equation} 
\label{eq8}
\bigwedge_2^1 R_z(\varphi_{1,2}) = \begin{pmatrix}
1 & 0 & 0 & 0 \\
0 & 1 & 0 & 0 \\
0 & 0 & 1 & 0 \\
0 & 0 & 0 & e^{i \varphi_{1,2}}
\end{pmatrix}.
\end{equation}

This controlled rotation establishes the longitudinal component of our third Bloch sphere.

{\bf Entanglement in the amplitude component}

Having established the phase parameterization, we now define the behavior of the latitudinal component of the third sphere, represented by the angle $\theta_{1,2}$.

In a separable state, the moduli of the amplitudes are strictly determined by the tensor product of the individual qubit spheres. To represent a general entangled state, we require a mechanism to redistribute the probability weights among the basis states, specifically between $\ket{10}$ and $\ket{11}$.

This redistribution is mathematically analogous to the function of a CNOT gate, which is the standard mechanism for generating maximally entangled Bell states (e.g. $\ket{\beta_{00}} = \frac{\sqrt 2}{2} \ket{00} + \frac{\sqrt 2}{2} \ket{11}$, as shown in Figure \ref{fig:BellS})

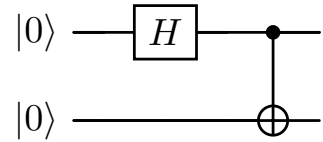
\begin{figure}[htbp]
    \centering
    \scalebox{1.3}{
    \begin{quantikz}[column sep=0.8em, row sep=1.2em]
        \lstick{$\ket{0}$} & \qw & \gate{H} & \qw & \ctrl{1} & \qw \\
        \lstick{$\ket{0}$} & \qw & \qw      & \qw & \targ    & \qw &
    \end{quantikz}
    }
    \caption{Bell State}
    \label{fig:BellS}
\end{figure}

It is evident from Eq. \eqref{eq6} that no combination of independent single-qubit angles can yield a Bell state via tensor product. The CNOT gate achieves this by interchanging the amplitudes of the $\ket{01}$ and $\ket{11}$ basis vectors. To allow for a continuous and complete redistribution capable of reaching any point in the Hilbert space, we generalize this action by introducing a controlled $Y$-rotation, denoted as $\bigwedge_2^1 R_y(\theta_{1,2})$:

\begin{equation} 
\label{eq10}
\bigwedge_2^1 R_y(\theta_{1,2}) = \begin{pmatrix}
1 & 0 & 0 & 0 \\
0 & 1 & 0 & 0 \\
0 & 0 & \cos \theta_{1,2} & - \sin \theta_{1,2} \\
0 & 0 & \sin \theta_{1,2} & \cos \theta_{1,2} 
\end{pmatrix} .
\end{equation}

By hierarchically applying these controlled transformations to the separable tensor product state $\ket{Q_1} \otimes \ket{Q_2}$ (we will see this hierarchy in detail in the next section), we can fully parameterize any general two-qubit quantum state using our proposed three-sphere model, defined by the parameters $\{(\theta_1, \varphi_1), (\theta_2, \varphi_2),(\theta_{1,2}, \varphi_{1,2}) \}$.

The complete state is given by:

\begin{small} 
\begin{align} 
\label{eq11}
&\ket{Q_1 Q_2} = \bigwedge_2^1 R_z(\varphi_{1,2}) \cdot \bigwedge_2^1 R_y(\theta_{1,2}) \cdot (\ket{Q_1} \otimes \ket{Q_2}) \nonumber \\  
&= \bigwedge_2^1 R_z(\varphi_{1,2}) \cdot \bigwedge_2^1 R_y(\theta_{1,2}) \cdot \bigg[ \bigg( \cos \frac{\theta_1}{2} \ket{0}+  \nonumber \\
&\quad+ e^{i \varphi_1} \sin \frac{\theta_1}{2} \ket{1} \bigg) \otimes \left( \cos \frac{\theta_2}{2} \ket{0} + e^{i \varphi_2} \sin \frac{\theta_2}{2} \ket{1} \right) \bigg] \nonumber \\
&= \bigwedge_2^1 R_z(\varphi_{1,2}) \cdot \bigwedge_2^1 R_y(\theta_{1,2}) \cdot \bigg[ \cos \frac{\theta_1}{2} \cos \frac{\theta_2}{2} \ket{00} +\nonumber \\
&\quad + \cos \frac{\theta_1}{2} \sin \frac{\theta_2}{2}  e^{i \varphi_2} \ket{01} + \sin \frac{\theta_1}{2} \cos \frac{\theta_2}{2}  e^{i \varphi_1} \ket{10}+ \nonumber \\
&\quad + \sin \frac{\theta_1}{2} \sin \frac{\theta_2}{2}  e^{i (\varphi_1 + \varphi_2)} \ket{11} \bigg] .
\end{align}
\end{small}

As this expression shows, applying these controlled transformations to the separable state yields the amplitudes and phases derived in Theorem \ref{th1}. Just as the single qubit state is represented geometrically by the two angles of a standard Bloch sphere, the bipartite entanglement is similarly mapped into the system via controlled rotations encoded by two additional angles

Observing the arguments in Eq. \eqref{eq11}, we can see that there is an asymmetry between the latitudinal angles of our model: those representing single-qubit spheres ($\theta_1$ and $\theta_2$) carry a $1/2$ factor, whereas the one corresponding to the sphere encoding the entanglement ($\theta_{1,2}$) does not. This is because the entanglement parameter does not map to a single orthogonal state but rather acts as a mixing parameter that must be capable of completely redistributing the amplitudes among the entangled states.

In Eq. \eqref{eq11}, we construct the representation by first taking the tensor product of two separate qubits and subsequently applying controlled rotations. While this is intuitive for understanding the physical introduction of entanglement, it presents a significant mathematical drawback for scaling the model. Applying the tensor product early distributes the single-qubit phases among the vector components, which significantly complicates the algebraic structure of the controlled rotations when scaling to 3 or more qubits. To resolve this and establish a framework capable of scaling to $N$ qubits, we must separate the amplitude generation from the phase assignment. We can reorganize our quantum circuit by applying all latitudinal rotations (around the $Y$-axis) first, followed by all longitudinal rotations (around the $Z$-axis). This grouping avoids algebraic mixing and leads to the definition of a unique normal form for the two-qubit system.

\begin{corollary}
    Given a general quantum system of two qubits, there exists a unique normal form to represent it by means of three Bloch spheres, systematically separating amplitude and phase operations:
\begin{small}
\begin{align} 
\label{FN2}
&\ket{Q_1 Q_2} = \bigwedge_2^1 R_z(\varphi_{1,2}) \cdot \bigwedge_1^\emptyset R_z(\varphi_1) \cdot \bigwedge_2^\emptyset R_z(\varphi_2) \nonumber \\
&\quad \cdot \bigwedge_2^1 R_y(\theta_{1,2}) \cdot \bigwedge_1^\emptyset R_y\left(\frac{\theta_1}{2}\right) \cdot \bigwedge_2^\emptyset R_y\left(\frac{\theta_2}{2}\right) \cdot \ket{00} \nonumber \\
&= \cos \frac{\theta_1}{2} \cdot \cos \frac{\theta_2}{2} \ket{00} + \cos \frac{\theta_1}{2} \cdot \sin \frac{\theta_2}{2} \cdot e^{i \varphi_2} \ket{01}+ \nonumber \\
&\quad + \sin \frac{\theta_1}{2} \cdot \cos \left(\frac{\theta_2}{2} + \theta_{1,2}\right) \cdot e^{i \varphi_1} \ket{10}+ \nonumber \\
&\quad + \sin \frac{\theta_1}{2} \cdot \sin \left(\frac{\theta_2}{2} + \theta_{1,2}\right) \cdot e^{i (\varphi_1 + \varphi_2 + \varphi_{1,2})} \cdot \ket{11} .
\end{align}
\end{small}
\end{corollary}

\begin{proof}
    This follows directly from Eq. \eqref{eq11} and the bijection established in Theorem 1. By leveraging the commutation properties of the controlled rotation gates (which we will formalize in the next section) we can safely reorder the operations into this strict hierarchy without altering the final quantum state. 

\end{proof}

\bigskip

\subsection{Modelling the 3-qubit system.
}
For a 3-qubit system we have three possibilities: No entanglement i.e. separable states, entanglement only between two out of three qubits and all the three qubits entangled.

In the case in which we have no entanglement at all, our model will be only 3 Bloch spheres, one for each qubit, and therefore the total state will be:

\begin{small}
\begin{align*} 
&\ket{\psi} = \bigwedge_1^\emptyset R_z(\varphi_1)\bigwedge_2^\emptyset R_z(\varphi_2)\bigwedge_3^\emptyset R_z(\varphi_3) \\
&\qquad \cdot \bigwedge_1^\emptyset R_y\left(\frac{\theta_1}{2}\right) \bigwedge_2^\emptyset R_y\left(\frac{\theta_2}{2}\right)\bigwedge_3^\emptyset R_y\left(\frac{\theta_3}{2}\right)\ket{000} \\
&\quad = \cos{\frac{\theta_1}{2}} \cdot \cos{\frac{\theta_2}{2}} \cdot \cos{\frac{\theta_3}{2}}\ket{000}+ \\
&\qquad + \cos{\frac{\theta_1}{2}} \cdot \cos{\frac{\theta_2}{2}} \cdot \sin{\frac{\theta_3}{2}} \cdot  e^{i\varphi_3}\ket{001}+ \\
&\qquad + \cos{\frac{\theta_1}{2}} \cdot \sin{\frac{\theta_2}{2}} \cdot \cos{\frac{\theta_3}{2}} \cdot e^{i\varphi_2}\ket{010}+ \\
&\qquad + \cos{\frac{\theta_1}{2}} \cdot \sin{\frac{\theta_2}{2}} \cdot \cos{\frac{\theta_3}{2}} \cdot e^{i(\varphi_2+\varphi_3)}\ket{011}+ \\
&\qquad + \sin{\frac{\theta_1}{2}} \cdot \cos{\frac{\theta_2}{2}} \cdot \cos{\frac{\theta_3}{2}} \cdot e^{i\varphi_1}\ket{100}+ \\
&\qquad + \sin{\frac{\theta_1}{2}} \cdot \cos{\frac{\theta_2}{2}} \cdot \sin{\frac{\theta_3}{2}} \cdot e^{i(\varphi_1+\varphi_3)}\ket{101}+ \\
&\qquad + \sin{\frac{\theta_1}{2}} \cdot \sin{\frac{\theta_2}{2}} \cdot \cos{\frac{\theta_3}{2}} \cdot e^{i(\varphi_1+\varphi_2)}\ket{110}+ \\
&\qquad + \sin{\frac{\theta_1}{2}} \cdot \sin{\frac{\theta_2}{2}} \cdot \sin{\frac{\theta_3}{2}} \cdot e^{i(\varphi_1+\varphi_2+\varphi_3)}\ket{111} .
\end{align*}
\end{small}

Here, the rotation matrices are now of size $8 \times 8$.

For a general 3-qubit pure state, we need 14 free parameters, therefore 7 different Bloch spheres. If we follow our previous procedure we end up with:

\begin{small}
\begin{align} 
\label{eq3esf_wrong}
&\ket{Q_1 Q_2 Q_3} = \bigwedge_2^1 R_z(\varphi_{1,2}) \bigwedge_3^2 R_z(\varphi_{2,3}) \bigwedge_3^1 R_z(\varphi_{1,3}) \nonumber \\
&\quad \cdot \bigwedge_1^\emptyset R_z(\varphi_1) \bigwedge_2^\emptyset R_z(\varphi_2) \bigwedge_3^\emptyset R_z(\varphi_3)\cdot \bigwedge_2^1 R_y(\theta_{1,2}) \nonumber \\
&\quad \cdot\bigwedge_3^2 R_y(\theta_{2,3}) \bigwedge_3^1 R_y(\theta_{1,3})\cdot \bigwedge_1^\emptyset R_y\left(\frac{\theta_1}{2}\right) \nonumber \\
&\quad  \cdot \bigwedge_2^\emptyset R_y\left(\frac{\theta_2}{2}\right) \bigwedge_3^\emptyset R_y\left(\frac{\theta_3}{2}\right) \ket{000}.
\end{align}
\end{small}

As in the two-qubit system, the $Y$-rotation angle corresponding to a ``qubit sphere'' must have a $1/2$ factor, while the angles corresponding to ``entanglement spheres'' do not. However, here we only have 6 spheres: the three spheres corresponding to each qubit and three additional spheres corresponding to the bipartite entanglement (qubit 1 with qubit 2, qubit 2 with qubit 3 and qubit 1 with qubit 3). Here we are only taking into account the entanglement between pairs of qubits. One sphere is missing: we must take into account the tripartite entanglement by introducing rotations on the third qubit controlled by the other two:

\begin{small}
\begin{align} 
\label{eq3esf_good}
&\ket{Q_1 Q_2 Q_3} = \bigwedge_3^{\{1,2\}} R_z(\varphi_{\{1,2\},3}) \bigwedge_3^2 R_z(\varphi_{2,3}) \bigwedge_3^1 R_z(\varphi_{1,3}) \nonumber \\
&\quad \cdot \bigwedge_2^1 R_z(\varphi_{1,2}) \bigwedge_1^\emptyset R_z(\varphi_1) \bigwedge_2^\emptyset R_z(\varphi_2) \bigwedge_3^\emptyset R_z(\varphi_3) \nonumber \\
&\quad \cdot \bigwedge_3^{\{1,2\}} R_y(\theta_{\{1,2\},3}) \bigwedge_3^2 R_y(\theta_{2,3}) \bigwedge_3^1 R_y(\theta_{1,3}) \nonumber \\
&\quad \cdot \bigwedge_2^1 R_y(\theta_{1,2}) \bigwedge_1^\emptyset R_y\left(\frac{\theta_1}{2}\right) \bigwedge_2^\emptyset R_y\left(\frac{\theta_2}{2}\right)  \nonumber \\
&\quad \cdot \bigwedge_3^\emptyset R_y\left(\frac{\theta_3}{2}\right) \ket{000} \nonumber \\
&= \sum_{xyz\in\{0,1\}} A_{xyz}e^{i\Phi_{xyz}}\ket{xyz} ,
\end{align}
\end{small}

where the coefficients are in the tables \ref{tab:coef} and \ref{tab:fases}.

\begin{table}[ht]
    \centering
    \renewcommand{\arraystretch}{0.5}
    \begin{tabular}{cc}
        \toprule
        \toprule
        \textbf{$\ket{xyz}$} & \textbf{$A_{xyz}$} \\
        \midrule
        
        000 & $\cos(\theta_1/2)\cos(\theta_2/2)\cos(\theta_3/2)$  \\
        \midrule
        
        001 & $\cos(\theta_1/2)\cos(\theta_2/2)\sin(\theta_3/2)$  \\
        \midrule
        
        010 & $\cos(\theta_1/2)\sin(\theta_2/2)\cos{(\theta_3/2+\theta_{2,3})}$  \\
        \midrule
        
        011 & $\cos(\theta_1/2)\sin(\theta_2/2)\sin{(\theta_3/2+\theta_{2,3})}$  \\
        \midrule
        
        100 & $\sin(\theta_1/2)\cos(\theta_2/2+\theta_{1,2})\cos{(\theta_3/2+\theta_{1,3})}$  \\
        \midrule
        
        101 & $\sin(\theta_1/2)\cos(\theta_2/2+\theta_{1,2})\sin{(\theta_3/2+\theta_{1,3})}$  \\
        \midrule
        
        110 & $\begin{aligned} &\sin(\theta_1/2)\sin(\theta_2/2+\theta_{1,2}) \\[-0.5ex] &\cdot \cos{(\theta_3/2+\theta_{1,3}+\theta_{2,3}+\theta_{\{1,2\}3})} \end{aligned}$  \\
        \midrule
        
        111 & $\begin{aligned} &\sin(\theta_1/2)\sin(\theta_2/2+\theta_{1,2}) \\[-0.5ex] &\cdot \sin{(\theta_3/2+\theta_{1,3}+\theta_{2,3}+\theta_{\{1,2\}3})} \end{aligned}$  \\
        
        \bottomrule
        \bottomrule
    \end{tabular}
    \caption{Moduli coefficients of the general expression for the 3 qubit case.}
    \label{tab:coef}
\end{table}

\begin{table}[ht]
    \centering
    \renewcommand{\arraystretch}{0.5}
    \begin{tabular}{cc}
        \toprule
        \toprule
        \textbf{$\ket{xyz}$} & \textbf{$\Phi_{xyz}$} \\
        \midrule
        
        000 & $0$ \\
        \midrule
        
        001 &  $\varphi_3$ \\
        \midrule
        
        010 &  $\varphi_2$ \\
        \midrule
        
        011 &  $\varphi_{2,3}+\varphi_3+\varphi_2$ \\
        \midrule
        
        100 &  $\varphi_1$ \\
        \midrule
        
        101 &  $\varphi_{1,3}+\varphi_3+\varphi_1$ \\
        \midrule
        
        110 &  $\varphi_{1,2}+\varphi_2+\varphi_1$ \\
        \midrule
        
        111 &  $\varphi_{\{12\},3}+\varphi_{2,3}+\varphi_{1,3}+\varphi_{1,2}  +\varphi_3+\varphi_2+\varphi_1$ \\
        
        \bottomrule
        \bottomrule
    \end{tabular}
    \caption{Phase coefficients of the general expression for the 3 qubit case.}
    \label{tab:fases}
\end{table}

We are now in a position to infer certain properties of our equation that will be crucial when generalizing our sphere parameterization to an $N$-qubit system. Let us consider each basis vector as a 3-bit word:

\begin{enumerate}

    \item The amplitude of each basis vector is given by a product of sines and cosines, matching the number of qubits (3 in this case). If a bit in the basis vector's binary word is 1, the corresponding function is a sine; if it is 0, the function is a cosine. For example, for the vector $\ket{010}$, the amplitude follows a ``$\cos\sin\cos$'' pattern, whereas for $\ket{111}$, it is ``$\sin\sin\sin$''.

    \item The argument for each of these trigonometric functions is constructed as follows: first, the angle corresponding to its specific qubit includes a $1/2$ factor, as it represents a single-qubit sphere rather than an entanglement sphere. To this angle, we add the angles corresponding to rotations controlled by any preceding bits in the binary word that are set to 1, along with any relevant combinations of these preceding 1-bits. Because these added angles represent entanglement spheres, they do not carry the $1/2$ factor.
    
    For instance, in the case of $\ket{010}$, the amplitude is:
    \begin{equation*}
        \cos(\theta_1/2)\sin(\theta_2/2)\cos{(\theta_3/2+\theta_{2,3})}.
    \end{equation*}
    
    Neither the first nor the second bit is preceded by a 1, so their arguments are simply the angles of their respective single-qubit spheres, $\theta_1/2$ and $\theta_2/2$. However, the third bit is preceded by a 1 (the second bit), so its argument must include both its single-qubit sphere angle ($\theta_3/2$) and the rotation controlled by the second bit ($\theta_{2,3}$).
    
    Conversely, for the $\ket{111}$ case, we have:
    \begin{align*}
        &\sin(\theta_1/2)\sin(\theta_2/2+\theta_{1,2}) \\
        &\quad \cdot \sin(\theta_3/2+\theta_{1,3}+\theta_{2,3}+\theta_{\{1,2\}3}).
    \end{align*}
    
    Here, since all three bits are 1, the argument for the third bit consists of its single-qubit sphere angle ($\theta_3/2$), followed by the rotations controlled by each of the preceding qubits set to 1 ($\theta_{1,3}$ and $\theta_{2,3}$), and  the combination of the two preceding 1-bits ($\theta_{\{1,2\}3}$).

    \item Finally, regarding the phase angle, there will be a single overall phase for each basis vector. This phase is the sum of the phase angles from the single-qubit spheres of the bits set to 1, plus the phase angles from the controlled rotations where the control bits (which are the ones set to 1) act on target bits (which are also set to 1), including their corresponding combinations. (We do not add controlled rotations if the target bit is not set to 1). 
    
    Continuing with the previous examples: for $\ket{010}$, the phase angle is simply $\varphi_2$. For $\ket{111}$, the phase angle is the sum $\varphi_{\{1,2\},3}+\varphi_{2,3}+\varphi_{1,3}+\varphi_{1,2}+\varphi_3+\varphi_2+\varphi_1$.

\end{enumerate}

To ensure the angles combine in this specific manner, we cannot apply the rotations to the $\ket{000}$ state arbitrarily, as not all of these rotations commute.

It is trivial to see that phase rotations ($R_z$) always commute with each other, as they are diagonal, regardless of their control and target qubits:

\begin{equation}
\label{eq:commuteZ}
\left[\bigwedge_i^J R_z(\varphi_{J,i}), \bigwedge_{i'}^{J'} R_z(\varphi_{J',i'})\right] = 0.
\end{equation}

However, controlled amplitude rotations ($R_y$) do not generally commute. Specifically, two controlled $R_y$ gates fail to commute if the target qubit of one gate acts as the control qubit for the other ($i \in J'$ or $i' \in J$), for example: 

\begin{align}
\label{eq:notcommuteY}
\left[\bigwedge_2^1 R_y(\theta_{1,2}), \bigwedge_3^2 R_y(\theta_{2,3})\right] &\neq 0 \\
\left[\bigwedge_1^\emptyset R_y(\theta_{1}), \bigwedge_2^1 R_y(\theta_{1,2})\right] &\neq 0. \nonumber
\end{align}

Furthermore, rotations around different axes acting on the same target qubit also do not commute 

\begin{equation}
    \left[\bigwedge_i^J R_z(\theta_{j,i}), \bigwedge_i^{J'} R_y(\theta_{j',i})\right] \neq 0.
\end{equation}

These commutation properties impose a strict hierarchical execution order in our Normal Form to avoid algebraic mixing of the amplitudes.
Therefore, all $R_y$ operations must be swept from the first qubit to the last (acting as targets before they act as controls), followed independently by the $R_z$ phase operations.

\begin{theorem}\label{th_3qubits}
    Given a general quantum system of $3$ qubits, there exists a unique normal form to represent its pure states by means of $7$ Bloch spheres, given by equation \ref{eq3esf_good}.

\end{theorem}

\begin{proof}
    
We will proceed analogously to the proof of the 2-qubit case.

\bigskip
\noindent {\bf Part 1.} Let us assume that the amplitudes of our state vector are real numbers and we do not consider the phases. We have 8 real positive values such that $\sum_{x\in\{0,1\}^3} A_{x}^2 = 1$. In this initial stage, we further assume that there is no entanglement, meaning the state is strictly separable and we only need three free parameters i.e. the single-qubit sphere angles $\{ \theta_1, \theta_2, \theta_3 \}$.

Following the logic established for the two-qubit system, we focus on the first three amplitudes to demonstrate the bijection for the separable components. The function connecting the Bloch sphere angles to these state vector amplitudes is:

\begin{equation*}
    f : [0, \pi]^3  \longrightarrow [0,1]^3
\end{equation*}

\begin{align*}
    f(\theta_1,\theta_2,\theta_3) &= (A_{000},A_{001},A_{010}) \\
    &= ( \cos(\theta_1/2) \cos(\theta_2/2) \cos(\theta_3/2) , \\
    &\qquad \cos(\theta_1/2) \cos(\theta_2/2) \sin(\theta_3/2) , \\
    &\qquad \cos(\theta_1/2) \sin(\theta_2/2) \cos(\theta_3/2) ).
\end{align*}

To verify that this function is bijective (excluding some points such as $\theta_1 = \pi$ or $\theta_2 = \pi$ that we will discuss later), we will show that we can uniquely recover the three angles from the three given amplitudes by reversing the trigonometric projections. We can find the angles sequentially as follows:

First, assuming $A_{000} \neq 0$ (the degenerate cases where a denominator is zero will be addressed at the end of this section), we divide $A_{001}$ by $A_{000}$
\begin{equation}
    \frac{A_{001}}{A_{000}} = \tan \frac{\theta_3}{2} \implies \theta_3 = 2 \arctan \left( \frac{A_{001}}{A_{000}} \right).
\end{equation}

Once $\theta_3$ is uniquely determined, we can isolate $\theta_2$ by dividing $A_{010}$ by $A_{000}$ and substituting the known $\theta_3$:
\begin{small}
\begin{equation}
    \frac{A_{010}}{A_{000}} = \frac{\tan(\frac{\theta_2}{2})}{\cos(\frac{\theta_3}{2})} \implies \theta_2 = 2 \arctan \left( \frac{A_{010} \cos(\frac{\theta_3}{2})}{A_{000}} \right).
\end{equation}
\end{small}

Finally, with $\theta_2$ and $\theta_3$ known, $\theta_1$ is trivially determined from the first amplitude:
\begin{equation}
    \theta_1 = 2 \arccos \left( \frac{A_{000}}{\cos \frac{\theta_2}{2} \cos \frac{\theta_3}{2}} \right).
\end{equation}

In cases where some amplitude is canceled (as in $\theta_1 = \pi$ or $\theta_2 = \pi$) and the value of other angles can be arbitrary, we will always take $\theta_i=0$ as canonical. This inverse mapping proves that any valid set of these three separable amplitudes corresponds to one, and only one, set of single-qubit angles $(\theta_1, \theta_2, \theta_3)$, thereby establishing the bijection for the non-entangled real subspace.

\medskip
\noindent {\bf Part 2.} We now introduce the bipartite and tripartite entanglement to our model, still assuming real amplitudes (ignoring phases). This entanglement is governed by the four remaining latitudinal parameters: $\{\theta_{1,2}, \theta_{1,3}, \theta_{2,3}, \theta_{\{1,2\},3}\}$. These angles control the redistribution of the probability weights among the remaining basis states via hierarchical controlled $Y$-rotations. Now that we no longer have a separable state, but a general one with entanglement, we need $2^3-1=7$ free parameters, which coincides with the number of angles we have.

Our goal is to show that we can uniquely determine these four entanglement angles from the remaining amplitudes, establishing the complete bijection for the moduli. Since the operations are applied in a strict hierarchical order, the dependence of the amplitudes on the angles forms a triangular structure, allowing us to solve them sequentially.

First, let us examine the amplitudes $A_{010}$ and $A_{011}$ from Table \ref{tab:coef}. They share the common factor $\cos(\theta_1/2)\sin(\theta_2/2)$. Assuming non-zero denominators, their ratio is independent of $\theta_1$ and $\theta_2$, depending only on the single-qubit angle $\theta_3$ and on the bipartite entanglement angle $\theta_{2,3}$,
\begin{equation}
    \frac{A_{011}}{A_{010}} = \tan \left( \frac{\theta_3}{2} + \theta_{2,3} \right).
\end{equation}
Operating, we obtain
\begin{equation}
    \theta_{2,3} = \arctan \left( \frac{A_{011}}{A_{010}} \right) - \frac{\theta_3}{2}.
\end{equation}
Since $\theta_3$ was already uniquely determined in Part 1, $\theta_{2,3}$ is completely fixed (note that $\theta_3$ was derived strictly from the ratio between $A_{000}$ and $A_{001}$ and therefore, it remains independent of the latitudinal angles of the entanglement spheres). As $\theta_{2,3} \in [0, \pi]$, the argument of the tangent sweeps the interval $[0,3\pi/2]$, the ratio can be negative. We can take care of this sign issue in the same way that we did with 2-qubits, with the phase.

Next, we address the bipartite entanglement controlled by the first qubit, $\theta_{1,2}$. To isolate it, we group the terms where the first qubit is in state $\ket{1}$:
\begin{align}
   A_{100}^2 + A_{101}^2 =  \sin^2 \frac{\theta_1}{2} \cos^2 \left( \frac{\theta_2}{2} + \theta_{1,2} \right) \nonumber \\
   A_{110}^2 + A_{111}^2 =  \sin^2 \frac{\theta_1}{2} \sin^2 \left( \frac{\theta_2}{2} + \theta_{1,2} \right).
\end{align}
The ratio isolates $\theta_{1,2}$:
\begin{equation}
    \frac{A_{100}^2 + A_{101}^2}{A_{110}^2 + A_{111}^2} =  \tan^2 \left( \frac{\theta_2}{2} + \theta_{1,2} \right) .
\end{equation}
With $\theta_2$ known from Part 1, we can uniquely solve for $\theta_{1,2}$. Again, the combined argument spans up to $3\pi/2$, making the tangent negative at some points. This would mean that some of the amplitudes are actually imaginary, which seems to contradict our initial assumption. However, as with the signs, we can convert them to real values if needed, along with the phases.

Taking the ratio of $A_{101}$ and $A_{100}$,
\begin{small}
\begin{equation}
    \frac{A_{101}}{A_{100}} = \tan \left( \frac{\theta_3}{2} + \theta_{1,3} \right), 
\end{equation}
\end{small}
we obtain $\theta_{1,3}$:
\begin{equation}
    \theta_{1,3} = \arctan \left( \frac{A_{101}}{A_{100}} \right) - \frac{\theta_3}{2}.
\end{equation}

Finally, the tripartite entanglement angle $\theta_{\{1,2\},3}$ appears only in the amplitudes $A_{110}$ and $A_{111}$. Their ratio encompasses all previously calculated angles:
\begin{equation}
    \frac{A_{111}}{A_{110}} = \tan \left( \frac{\theta_3}{2} + \theta_{1,3} + \theta_{2,3} + \theta_{\{1,2\},3} \right).
\end{equation}
Since all terms in the argument are already uniquely known except $\theta_{\{1,2\},3}$, we can trivially isolate it. 

In the last two calculations we have the same problem with the sign of the amplitudes which, as in all previous cases, can be solved by correctly adjusting the phases. This systematic unraveling mathematically proves that any arbitrary set of real, normalized 3-qubit amplitudes corresponds bijectively to a unique configuration of the 7 latitudinal angles $(\theta_1, \theta_2, \theta_3, \theta_{1,2}, \theta_{1,3}, \theta_{2,3}, \theta_{\{1,2\},3})$.

\medskip
\noindent{\bf Part 3.} Finally, we must demonstrate the bijection for the phase angles. Now our amplitudes are complex and can be expressed as $A_{xyz}e^{i\Phi_{xyz}}$, where $A_{xyz}$ is real and positive and $\Phi_{xyz}\in[0,2\pi)$. In parts 1 and 2 we saw that when recovering the angles we could end up with negative or even imaginary amplitude in a specific $\ket{xyz}$. Now we can fix that problem with a phase change in the corresponding $\Phi_{xyz}$.

The global phase is factored out such that the canonical reference phase is $\Phi_{000}$, leaving us with 7 independent relative phases that correspond in number to our 7 phase angles from our sphere model $\{ \varphi_1, \varphi_2, \varphi_3, \varphi_{1,2}, \varphi_{1,3}, \varphi_{2,3}, \varphi_{\{1,2\},3} \}$. It is easy to see that the matrix relating the sphere's angles $\varphi$ to the relative phases $\Phi$ is triangular with a non-zero diagonal. Since its determinant is non-zero, the system of equations is consistent and determinate. That is, any arbitrary phase distribution can be uniquely constructed. We compute the resulting phase values modulo $2\pi$ to ensure they remain within the $[0, 2\pi)$ interval. This step completes the proof, confirming the full bijective mapping between the 14 real parameters of the general 3-qubit Hilbert space and the 7 Bloch spheres.

\end{proof}

\subsection{Modelling the N-qubit system.}
Let $N$ be the number of qubits of our system and let $a$ be either the symbol $y$ or $z$. For any target qubit $i\in \{2,\dots,N\}$ and any number of control qubits $0\le k\le i-1$, we define $\mathcal{J}_{k,i}$ as the ordered sequence of all subsets of $\{1,\dots,i-1\}$ that have cardinality $k$. The number of elements of this sequence is $\binom{i-1}{k}$. We denote the $h$-th element of this sequence as $J^{(k)}_h$, so $\mathcal{J}_{k,i} = (J^{(k)}_1, J^{(k)}_2, \dots, J^{(k)}_M)$.

It follows that, if $k=0$, the sequence contains only one element, the empty set, thus $\mathcal{J}_{0,i}=(\emptyset)$ and, in this case, we consider $i \in \{1,\dots,N\}$. The ordering of the sequence is such that if $1 \leq u < v \leq \binom {i-1}{k}$ then $J^{(k)}_u$ precedes $J^{(k)}_v$ lexicographically. We define

\begin{equation}
     U_{k,i}^{a} = \prod_{h=1}^{\binom{i-1}{k}} \bigwedge_{i}^{J^{(k)}_h} R_a(\alpha_{J^{(k)}_h, i}),
    \label{Eq:U^a_k,i}
\end{equation}

where $\alpha_{J,i}$ is the corresponding rotation angle parameter, defined for each axis as follows:
\begin{small}
\begin{equation*}
    \text{For } a = y:  \alpha_{J^{(k)}_h,i} = \begin{cases} \theta_i/2 & \text{if } J^{(k)}_h = \emptyset \quad (k=0) \\ \theta_{J^{(k)}_h,i} & \text{if } J^{(k)}_h \neq \emptyset \quad (k \geq 1) \end{cases}.
\end{equation*}
\end{small}

\begin{small}
\begin{equation*}
    \text{For } a = z:  \alpha_{J^{(k)}_h,i} = \begin{cases} \varphi_i & \text{if } J^{(k)}_h = \emptyset \quad (k=0) \\ \varphi_{J^{(k)}_h,i} & \text{if } J^{(k)}_h \neq \emptyset \quad (k \geq 1) \end{cases}.
\end{equation*}
\end{small}

It is important to note that, in eq. \eqref{Eq:U^a_k,i},  the product symbol $\prod$ respects the ordered nature of the sequence $\mathcal{J}_{k,i}$ from left to right.

With these unitary operators defined, we can now establish the general expression that encodes the entire hierarchical structure for an arbitrary number of qubits.

\begin{theorem}\label{th_nqubits_3}
    Given a general quantum system of $N$ qubits, there exists a unique normal form to represent its pure states by means of $2^N-1$ Bloch spheres, given by the expression:

\begin{align}
    \ket{Q_1 \dots Q_N} &= \left(\prod^2_{i=N}\prod^1_{k=i-1}U^z_{k,i}\right)\left(\prod_{i=1}^NU^z_{0,i}\right) \nonumber \\
    &\quad \cdot \left(\prod^2_{i=N}\prod^1_{k=i-1}U^y_{k,i}\right)\left(\prod_{i=1}^NU^y_{0,i}\right)\cdot\ket{0\dots0}.
    \label{eq:GeneralExpression}
\end{align}
\end{theorem}

A rigorous mathematical demonstration of this theorem, based on an inductive proof using algorithmic state transformations, is detailed in the appendix. Below, we provide a constructive physical intuition relying on the interpretation of our model as controlled logic gates.

\begin{proof}
    To prove this, we will rely on the interpretation of our model as controlled logic gates. As can be seen in Figure \ref{fig:state_creation}, by following equation \eqref{eq:GeneralExpression}, we obtain a circuit where adding a new qubit entails implementing a block of controlled rotations on this qubit.

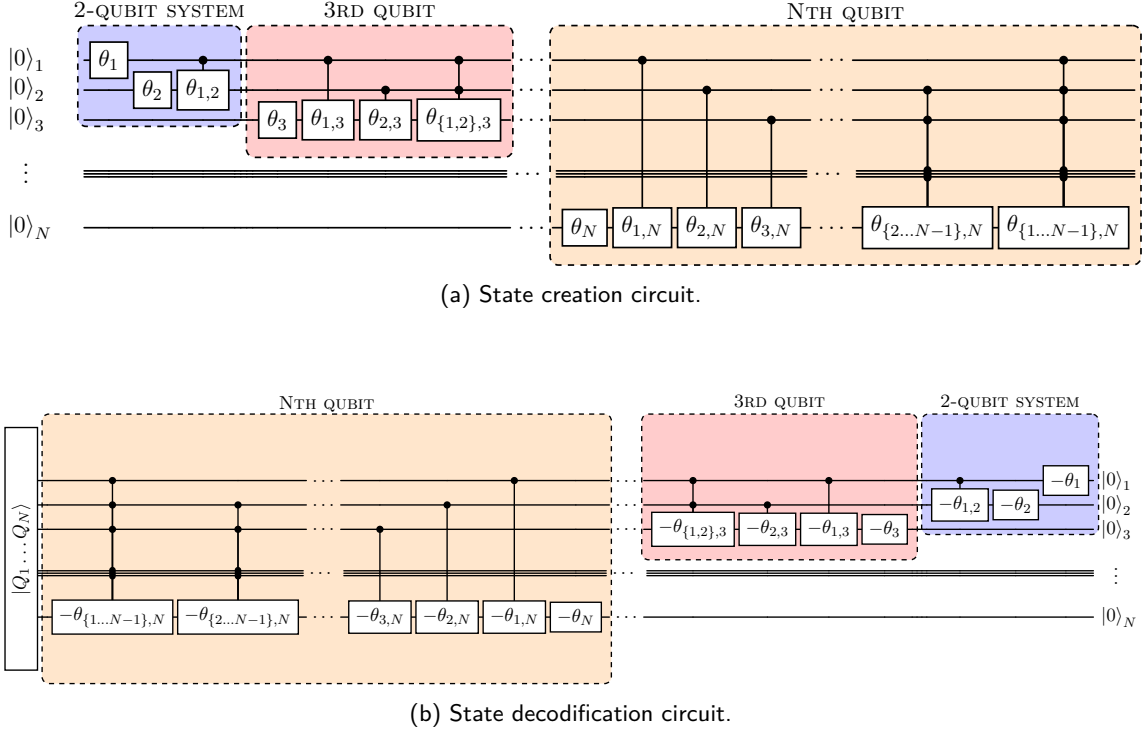
\begin{figure*}[tbp] % El asterisco (*) hace que la figura ocupe ambas columnas
    \centering
    
    % --- Subfigura 1: Creación ---
    \begin{subfigure}{\textwidth}
        \centering
        \resizebox{0.9\textwidth}{!}{
        \begin{quantikz}[row sep={0.5cm, between origins}, column sep=0.1cm, wire types={q,q,q,b,q,q}, classical gap=0.05cm]
            \lstick{\rlap{$\ket{0}_1$}\phantom{$\ket{0}_{N-1}$}} & \gate{\theta_1}\gategroup[2,steps=3,style={dashed,rounded corners,fill=blue!20, inner xsep=2pt},background,label style={anchor=north,yshift=0.3cm}]{{\sc 2-qubit system}} && \ctrl{1} &&&&&\gategroup[3,steps=4,style={dashed,rounded corners,fill=red!20, inner xsep=2pt},background,label style={anchor=north,yshift=0.3cm}]{{\sc 3rd qubit}}& \ctrl{2}& &\ctrl{2} & \ \ldots\ &\gategroup[5,steps=7,style={dashed,rounded corners,fill=orange!20, inner xsep=2pt},background,label style={anchor=north,yshift=0.3cm}]{{\sc Nth qubit}}& \ctrl{4} &&&   \ \ldots\ & & \ctrl{4} & \\
            \lstick{\rlap{$\ket{0}_2$}\phantom{$\ket{0}_{N-1}$}} && \gate{\theta_2} & \gate{\theta_{1,2}} &&&&&&& \ctrl{1} & \ctrl{1} &\ \ldots\ &&& \ctrl{3} &&  \ \ldots\  & \ctrl{3}& \ctrl{3} & \\
            \lstick{\rlap{$\ket{0}_3$}\phantom{$\ket{0}_{N-1}$}} &&&&&&&& \gate{\theta_3} &\gate{\theta_{1,3}}& \gate{\theta_{2,3}} & \gate{\theta_{\{1,2\},3}} &\ \ldots\ &&&& \ctrl{2} & \ \ldots\  & \ctrl{2}& \ctrl{2} & \\[0.4cm]
            \lstick{\rlap{\raisebox{0.1cm}{\hspace{0.25cm}$\vdots$}}\phantom{$\ket{0}_{N-1}$}} &&&&&&&&&&&& \ \ldots\ \ &&&&& \ \ldots\ \ &\ctrl{1}&\ctrl{1}& \\[0.4cm]
            \lstick{\rlap{$\ket{0}_N$}\phantom{$\ket{0}_{N-1}$}} &&&&&&&&&&&& \ \ldots\ & \gate{\theta_N}& \gate{\theta_{1,N}} & \gate{\theta_{2,N}} & \gate{\theta_{3,N}} &  \ \ldots\  & \gate{\theta_{\{2\dots N-1\},N}} & \gate{\theta_{\{1\dots N-1\},N}}
        \end{quantikz}
        }
        \caption{State creation circuit.}
        \label{fig:state_creation}
    \end{subfigure}
    
    \vspace{1cm} % Espacio de separación entre los dos circuitos
    
    % --- Subfigura 2: Decodificación ---
    \begin{subfigure}{\textwidth}
        \centering
        \resizebox{0.9\textwidth}{!}{
        \begin{quantikz}[row sep={0.5cm, between origins}, column sep=0.1cm, wire types={q,q,q,b,q,q}, classical gap=0.05cm]
            \gate[5][0.5cm]{\rotatebox{90}{$\ket{Q_1 \dots Q_N}$}} &&& \ctrl{4}\gategroup[5,steps=7,style={dashed,rounded corners,fill=orange!20, inner xsep=2pt},background,label style={anchor=north,yshift=0.3cm}]{{\sc Nth qubit}} & & \ \ldots\ & & & \ctrl{4} & & \ \ldots\ & \ctrl{2}\gategroup[3,steps=4,style={dashed,rounded corners,fill=red!20, inner xsep=2pt},background,label style={anchor=north,yshift=0.3cm}]{{\sc 3rd qubit}} & & \ctrl{2} & & & & & & \ctrl{1}\gategroup[2,steps=3,style={dashed,rounded corners,fill=blue!20, inner xsep=2pt},background,label style={anchor=north,yshift=0.3cm}]{{\sc 2-qubit system}} & & \gate{-\theta_1} & \rstick{$\ket{0}_1$} \\
            &&& \ctrl{3} & \ctrl{3} & \ \ldots\ & & \ctrl{3} & & & \ \ldots\ & \ctrl{1} & \ctrl{1} & & & & & & & \gate{-\theta_{1,2}} & \gate{-\theta_2} & & \rstick{$\ket{0}_2$} \\
            &&& \ctrl{2} & \ctrl{2} & \ \ldots\ & \ctrl{2} & & & & \ \ldots\ & \gate{-\theta_{\{1,2\},3}} & \gate{-\theta_{2,3}} & \gate{-\theta_{1,3}} & \gate{-\theta_3} & & & & & & & & \rstick{$\ket{0}_3$} \\[0.4cm]
            &&& \ctrl{1}&\ctrl{1} & \ \ldots\ \ & & & & & \ \ldots\ \ & & & & & & & & & & & & \rstick{\raisebox{0.1cm}{\hspace{0.25cm}$\vdots$}} \\[0.4cm]
            &&& \gate{-\theta_{\{1\dots N-1\},N}} & \gate{-\theta_{\{2\dots N-1\},N}} & \ \ldots\ & \gate{-\theta_{3,N}} & \gate{-\theta_{2,N}} & \gate{-\theta_{1,N}} & \gate{-\theta_N} & \ \ldots\ & & & & & & & & & & & & \rstick{$\ket{0}_N$}
        \end{quantikz}
        }
        \caption{State decodification circuit.}
        \label{fig:state_decodification}
    \end{subfigure}
    
    \caption{Hierarchical representation of the model translated into controlled quantum logic gates for an arbitrary $N$-qubit system.The blocks labeled with $\theta$ parameters represent the corresponding rotation gates. For diagrammatic simplicity, we omitted that for single-qubit $Y$-rotations their angles have a $1/2$ factor, whereas the multipartite entanglement rotation angles remain unscaled. Furthermore, the complete state generation process involves applying this entire circuit structure twice, first using $Y$-rotations, followed by the same structure with $Z$-rotations.}
    \label{fig:combined_circuits}
\end{figure*}

Assuming this decomposition is correct and unique, from an arbitrary state $\ket{Q_1\dots Q_N}$, we can obtain the state $\ket{0\dots0}$ by undoing all the rotations one by one, that is, applying the gates in reverse order with negative angles, as shown in Figure \ref{fig:state_decodification}.

Based on this construction, we will demonstrate that our model is capable of uniquely generating any arbitrary state. To do this, we will apply the gates associated with the decoding of the $N$-th qubit (the orange block in Figure \ref{fig:state_decodification}) to our arbitrary state and extract the angles that cancel out all the state vector components where the $N$-th qubit is in the state $\ket{1}$. In this way, once we have applied the entire block in our decoding process, we will have reduced our problem of obtaining the parameters of an $N$-qubit state to an $(N-1)$-qubit problem. Using this methodology, we will eventually reach a 3-qubit system, for which we have already proven the validity of our model via Theorem \ref{th_3qubits}. Let us detail how this is performed specifically.
Let us consider the arbitrary state $\ket{Q_1\dots Q_N}=\sum_{x\in\{0,1\}^{N-1},y\in\{0,1\}}A_{xy}\ket{xy}e^{i\Phi_{xy}}$, where the binary string $x$ represents the state of the first $N-1$ qubits, $y\in\{0,1\}$, and $A_{xy}$ are the real amplitudes satisfying the normalization condition $\sum_{x\in\{0,1\}^{N-1},y\in\{0,1\}}A_{xy}^2=1$. In this first part, we will omit the phases and operate solely with the moduli and the Y-rotations. The state from which we must obtain the angles $\theta$ reduces to $\ket{Q_1\dots Q_N} = \sum_{x \in \{0,1\}^{N-1}, y \in \{0,1\}} A_{xy}\ket{xy}$.

We begin by applying the decoding block of the $N$-th qubit to our state. Since we are only applying the rotations associated with the $N$-th qubit, each of our gates acts on 2-dimensional subspaces defined by each string $x$. Therefore, in our resulting state, only the amplitudes with the same $x$ are mixed: $A_{x0}$ and $A_{x1}$. It can be easily verified that the resulting state is:

\small{
\begin{equation}
    \begin{pmatrix}
        \cos{(\Theta_{00\dots0})}A_{00\dots00}+\sin{(\Theta_{00\dots0})}A_{00\dots01}\\
        \cos{(\Theta_{00\dots0})}A_{00\dots01}-\sin{(\Theta_{00\dots0})}A_{00\dots00}\\
        \vdots\\
        \cos{(\Theta_x)}A_{x0}+\sin{(\Theta_x)}A_{x1}\\
        \cos{(\Theta_x)}A_{x1}-\sin{(\Theta_x)}A_{x0}\\
        \vdots\\
        \cos{(\Theta_{11\dots1})}A_{11\dots10}+\sin{(\Theta_{11\dots1})}A_{11\dots11}\\
        \cos{(\Theta_{11\dots1})}A_{11\dots11}-\sin{(\Theta_{11\dots1})}A_{11\dots10}\\
    \end{pmatrix},
\end{equation}
}

where we have defined the effective angle $\Theta_x$ that mixes the state vector components with the same $x$

\begin{equation}
    \Theta_x = \frac{\theta_N}{2} + \sum_{k=1}^{N-1} \sum_{h=1}^{\binom{N-1}{k}} \left( \prod_{j \in J^{(k)}_h} x_j \right) \theta_{J^{(k)}_h, N}.
    \label{eq:Eff_angle}
\end{equation}

If we equate to zero all the state vector components that have the $N$-th qubit in the state $\ket{1}$ (i.e., the new amplitudes $A'_{x1}$ of our vector), we find that the required effective angle is: $\Theta_x=\arctan(A_{x1}/A_{x0})$. Thanks to the structure of our equation \eqref{eq:Eff_angle}, which establishes a triangular relationship between the effective angles and the angles of our spheres, we can solve for them sequentially. To obtain all our parameters correctly, it is imperative that this resolution be done sequentially, as our effective angle combines the parameter of interest with multiple lower-order contributions. The order in which we obtain the angles of our model must be progressive, starting with the single-qubit angle and advancing according to the cardinality $k$ of the control subsets.

\begin{small}
\begin{align*}
    \theta_N &= 2\arctan\left(\frac{A_{00\dots01}}{A_{00\dots00}}\right) \\
    \theta_{J^{(1)}_h, N} &= \arctan\left(\frac{A_{x1}}{A_{x0}}\right) - \frac{\theta_N}{2} \\
    &\qquad \text{with } x_j = 1 \iff j \in J^{(1)}_h \\
    \theta_{J^{(2)}_h, N} &= \arctan\left(\frac{A_{x1}}{A_{x0}}\right) - \frac{\theta_N}{2} - \sum_{l: J^{(1)}_l \subset J^{(2)}_h} \theta_{J^{(1)}_l, N} \\
    &\qquad \text{with } x_j = 1 \iff j \in J^{(2)}_h \\
    &\vdots \\
    \theta_{J^{(k)}_h, N} &= \arctan\left(\frac{A_{x1}}{A_{x0}}\right) - \frac{\theta_N}{2} - \sum_{m=1}^{k-1} \sum_{l: J^{(m)}_l \subset J^{(k)}_h} \theta_{J^{(m)}_l, N} \\
    &\qquad \text{with } x_j = 1 \iff j \in J^{(k)}_h.
\end{align*}
\end{small}

To formalize this correctly without dragging along unknown terms, we have applied the condition $x_j = 1 \iff j \in J^{(k)}_h$. This allows us to activate only the bits in our string that correspond to the control qubits integrating the analyzed subset. Furthermore, since the summation subtracted from the arctangent is restricted to iterating from order $1$ to $k-1$, we ensure that all variables originate from previously resolved stages.

In order to maintain the uniqueness of our representation in cases where divergences arise, we must establish certain values as canonical: if $A_{x0}=0$ and $A_{x1}\neq 0$, we will take the corresponding angle to be $\pi/2$. Conversely, if both amplitudes are zero ($A_{x0}=A_{x1}=0$), we will take $0$ as the canonical value for the corresponding angle. At first glance, the appearance of these divergences might seem like a mathematical artifact of our equations. However, they arise naturally due to our geometric model. These situations occur when our state is located at one of the poles of any of our entanglement spheres. Just as in the traditional Bloch sphere, where the azimuthal coordinate loses its definition at the poles, our entanglement parameter becomes indeterminate. Therefore, establishing canonical values for the cases in which we encounter divergences does not constitute an \textit{ad-hoc} fix, but rather a means of maintaining the rigor and standardization of our model by taking its topology into account.

Once the block corresponding to the $N$-th qubit has been applied, any state vector component with the $N$-th qubit in the state $\ket{1}$ will be zero, leaving our system in the state:

\begin{equation}
    \left(\sum_{x\in\{0,1\}^{N-1}}A'_x\ket{x}\right)\otimes \ket{0}.
\end{equation}

That is, we have reduced the dimensionality of our problem from $N$ to $N-1$. This same procedure allows us to continue reducing the dimensionality of our problem until we reach the 3-qubit case, which was proven in Theorem \ref{th_3qubits}.

To obtain the phase angles, the procedure is identical, applying the gates in reverse order and with negative angles. In this case, we do not need to worry about divergences when extracting the angles.
\end{proof}

\section{Discussion}

Throughout this paper, we have presented a novel geometric model, based on the extension of the traditional Bloch sphere, to represent multi-qubit systems. We began by demonstrating that the parameterization of two-qubit systems using the latitudinal and longitudinal angles of three Bloch spheres is unique. Of these three, two describe each qubit individually, while the third encapsulates all the information regarding their entanglement. Furthermore, we have shown that this parameterization can be easily translated into a quantum circuit by applying uncontrolled rotations along the $Y$ and $Z$ axes for the angles of the ``qubit spheres,'' and controlled rotations along the same axes for the ``entanglement spheres.''

Building upon the two-qubit parameterization, we extended the model to 3-qubit systems, demonstrating once again that the correspondence between the state vector and our sphere model is bijective. We also established the foundational rules for the strict order in which the gates must be applied to achieve the required encoding. Due to the commutation relations of the different controlled gates, this order is not arbitrary. In the 3-qubit case, we have shown that 7 spheres are required in total: 3 ``qubit spheres'', 3 spheres for bipartite entanglement, and one final sphere encoding the tripartite entanglement.

Finally, relying on the hierarchy established for the 3-qubit case, we were able to generalize our model to $N$-qubit systems, proving that the representation uniquely maps the entire Hilbert space. While other geometric parameterizations found in the literature, such as the works by Wharton \cite{wharton} and Wie \cite{Wie}, add angles in an ad hoc manner or use non-unit spheres (in addition to only exploring the 2-qubit space), we have established a clear, well-structured geometric model capable of scaling to an arbitrary number of qubits without restrictions on entanglement.

From a computational science perspective, this parameterization offers a powerful alternative to the standard state vector, potentially reducing the computational cost of quantum simulations drastically when there is no entanglement between certain qubits. In the worst-case scenario, where all qubits are fully entangled, the number of required parameters equals that of the state vector. This provides a faster and more efficient method for storing and manipulating data, thereby increasing the capacity and speed of simulators that implement this parameterization. This advantage is particularly beneficial when simulating real quantum computers, where maintaining global entanglement across all qubits is currently practically impossible.

From a purely theoretical standpoint, this model lays the groundwork for a new framework to analyze quantum operations. Rather than being limited to representing elements of the state space, this approach allows quantum gates (elements of higher-order Lie groups such as SU(4), SU(8), and beyond) to be interpreted as geometric transformations and coupled rotations among these spheres. This perspective provides an intuitive and highly promising new avenue for studying quantum dynamics, multipartite entanglement, and circuit optimization.

\section{Acknowledgments}

This research received no specific grant from any funding agency in the public, commercial, or not-for-profit sectors.

\section{Author contributions}
All authors contributed equally to all aspects of this work, including the conceptualization, mathematical development, and writing of the manuscript. An AI language model was utilized exclusively to assist with LaTeX document formatting, typesetting, and resolving layout constraints. The AI was not used for the generation of scientific content, calculations, or theoretical derivations. All authors have reviewed and approved the final manuscript and assume full responsibility for its contents.

\section{Competing Interests}

The authors declare no competing interests.

\onecolumn
\appendix

\section{More formal proof of Theorem 3}

We will make use of the following notation. We will consider a $0 \leq x < 2^N$ sometimes as a number and sometimes as its binary representation; with $S(x)$ we will denote the subset of $\{1,2,\dots, N\}$ whose characteristic vector is $x$. For example, if $x =13$ then its binary representation is $x=(0,1,1,0,1)$ and  $S(x)=\{2,3,5\}$.  First we have the following Lemma.

 \begin{lemma} \label{le:subset}
    Let $\ket{\mathbf{Q}}$ be an $N$ qubits system, $m \in \{2,3,\dots, N\}$ and $J \subseteq \{1,2,$ $\dots, m-1\}$. 
    If $U$ is a single qubit gate then the components of $\bigwedge^{J}_m U \ket{\mathbf{Q}}$ with index $x$ such that $S(x) \nsubseteq J$ have the same values as the corresponding components of $\ket{\mathbf{Q}}$. 
 \end{lemma}

\begin{proof}
    
 By the properties of tensor product and by the definition of controlled operator, the components with index $x$ of $\ket{\mathbf{Q}}$ affected by $\bigwedge^{J}_m U$ are only the ones in  which $J \subseteq S(x)$. Therefore all the components with index $x$ of $\bigwedge^{J}_m U \ket{\mathbf{Q}}$ for which $J \nsubseteq S(x)$ will have the same value of the corresponding components of $\ket{\mathbf{Q}}$.  
\end{proof}

Note that in Lemma \ref{le:subset} in the the case in which $J= \emptyset$  the controlled operator $\bigwedge^{\emptyset}_m U \ket{\mathbf{Q}}$ generally act on every component of $\ket{\mathbf{Q}}$

\begin{lemma} \label{le:less_thanZ}
    Let $\ket{\mathbf{Q}}$ be an $N$ qubits system. Let $m \in \{1,2, \dots, N\}$, $J' \subseteq \{1,2, \dots, m\}$ such that $m \in J'$ and $k$ such that $S(k)=J'$.  If  $J = J' \setminus \{m\}$ then all the components of vector $$\bigwedge^{J}_m R_z(\varphi_{J,m}) \ket{\mathbf{Q}},$$ with index less than $k$ are equal to the corresponding components of $\ket{\mathbf{Q}}$.
\end{lemma}

\begin{proof}
    Let $\ket{\mathbf Q_1} = \bigwedge^{J}_m R_z(\varphi_{J,m}) \ket{\mathbf{Q}}$.  By Lemma \ref{le:subset}, the controlled operator $\bigwedge^{J}_m R_z(\varphi_{J,m})$ will act only on the components $x$ of $\ket{\mathbf Q}$ such that $J \subseteq S(x)$. However for those components whose index $x$ is such that $J \subseteq S(x)$ and $J' \nsubseteq S(x)$ the operator do not modify their values because $R_z(\varphi_{J,m})$ is diagonal and the first element of the diagonal is one. Therefore $\bigwedge^{J}_m R_z(\varphi_{J,m})$ will act only on those components with index $x$ such that $J' \subseteq S(x)$.  Now we have that for all $0\leq h<k$, $S(k) \nsubseteq S(h)$, for otherwise $k \leq h$. Thus $J' =S(k) \nsubseteq S(h)$ and the lemma follows by Lemma \ref{le:subset}.
\end{proof}

\bigskip

\begin{lemma}
    Let $\ket{\mathbf{Q}} =\ket{Q_1 \dots Q_N}$ be a state of $N$ qubits. Let $M = 2^N-1$ and 
$$\ket{\mathbf{Q}} =\sum_{x=0}^M A_xe^{i \phi_x} \ket{x}$$.
There exists a transformation $\mathbf U_z$
\begin{equation*}
    \mathbf U_z =\prod^N_{m=1}\prod^0_{k=m-1}U^z_{k,m}
\end{equation*}
such that if $\ket{\mathbf Q_1}=\mathbf U_z \ket{\mathbf{Q}} = \sum_{x=0}^M R_xe^{i\psi_x}\ket{x}$ then  $\psi_x = 0$ for all $0 \leq x \leq M$. 
\end{lemma}

\begin{proof}
    
Consider the procedure of Algorithm  \ref{alg:state_transformationZ}. 
\begin{figure}[htbp]
\caption{State Vector Transformation Procedure $Z$}
\label{alg:state_transformationZ}
\begin{algorithmic}[1]
\Require The state vector $\ket{\mathbf{Q}} =\sum_{x=0}^M A_xe^{i \phi_x} \ket{x}$
\Ensure The set of angles of rotation $\varphi_{J,m}$, $J \in \mathcal{J}_{k,m}$, $1 \leq m\leq N$, $0 \leq k <m $$\ket{\mathbf{Q_1}}$
\For{$x = 1 \to 2^{N}-1$}
    \State Let $J' \subseteq \{1,\dots, N\}$ such that $S(x)= J'$
    \State  Let $m= \max\{ j | j \in J'\}$
    \State Let $J = J' \setminus \{m\}$
    \State Set $\varphi_{J,m} \leftarrow -\phi_{x}$
    \State Set $\ket{\mathbf{Q}} \leftarrow \bigwedge^J_m R_z(\varphi_{J,m}) \cdot \ket{\mathbf{Q}}$  
    \label{step:transformZ}
    
\EndFor
\State $\ket{\mathbf{Q_1}} \leftarrow \ket{\mathbf{Q}}$
\State \Return $\varphi_{J,m}$, $J \in \mathcal{J}_{k,m}$, $1 \leq m\leq N$, $0 \leq k <m $
\end{algorithmic}
\end{figure}
Note that at step \ref{step:transformZ} of Algorithm \ref{alg:state_transformationZ}, the phase of component $x$ of $\ket{\mathbf{Q}}$ will be zero. Furthermore, by Lemma \ref{le:less_thanZ}, the controlled operator $\bigwedge^J_m R_z(\varphi_{J,m})$ will act only on components with index $h$ such that $h \geq x$. Therefore, at step \ref{step:transformZ} of Algorithm  \ref{alg:state_transformationZ} the component with index $x$ will have its phase equal to zero and it will be not modified during the rest of the execution of the algorithm.  Thus, at the end of the algorithm, all the  components of $\ket{\mathbf Q_1}$ have their phase equal to zero. The set of angles given in output therefore satisfy the requirement for the transformation $\mathbf U_z$.
\end{proof}

\bigskip

\begin{example}
    As an example, consider a system with three qubits. The tables \ref{tab:part_1} and \ref{tab:part_2} show, for each iteration $x$, the parameters of the algorithm and the phase of each component of $\ket{\mathbf Q}$. At the end of the algorithm the phase of each component of system $\ket{\mathbf Q_1}$ is zero.
\end{example}
\begin{table}[htbp]
\centering
\caption{The parameters and the phases of system $\ket{\mathbf Q}$ at each iteration $x = 1$ to $5$ of Algorithm \ref{alg:state_transformationZ}}
\label{tab:part_1}
\small
\begin{tabular}{l *{5}{c}}
\toprule
$x$ & 1 & 2 & 3 & 4 & 5 \\
$m$ & 3 & 2 & 3 & 1 & 3 \\
$J$ & $\emptyset$ & $\emptyset$ & $\{2\}$ & $\emptyset$ & $\{1\}$ \\
$\varphi_{J,m}$ & $-\phi_1$ & $-\phi_2$ & $-\phi_3+\phi_1+\phi_2$ & $-\phi_4$ & $-\phi_5+\phi_4+\phi_1$ \\
\midrule
$\ket{000}$ & 0 & 0 & 0 & 0 & 0 \\
$\ket{001}$ & 0 & 0 & 0 & 0 & 0 \\
$\ket{010}$ & $\phi_2$ & 0 & 0 & 0 & 0 \\
$\ket{011}$ & $\phi_3-\phi_1$ & $\phi_3-\phi_1-\phi_2$ & $0$ & 0 & 0 \\
$\ket{100}$ & $\phi_4$ & $\phi_4$ & $\phi_4$ & 0 & 0 \\
$\ket{101}$ & $\phi_5-\phi_1$ & $\phi_5-\phi_1$ & $\phi_5-\phi_1$ & $\phi_5-\phi_1-\phi_4$ & 0 \\
$\ket{110}$ & $\phi_6$ & $\phi_6-\phi_2$ & $\phi_6-\phi_2$ & $\phi_6-\phi_2-\phi_4$ & $\phi_6-\phi_2-\phi_4$ \\
$\ket{111}$ & $\phi_7-\phi_1$ & $\phi_7-\phi_1-\phi_2$ & $\phi_7-\phi_3$ & $\phi_7-\phi_3-\phi_4$ & $\phi_7-\phi_5-\phi_3+\phi_1$ \\
\bottomrule
\end{tabular}
\end{table}

\vspace{2em}

\begin{table}[htbp]
\centering
\caption{The parameters and the phases of system $\ket{\mathbf Q}$ at each iteration $x = 6$ to $7$ of Algorithm \ref{alg:state_transformationZ}}
\label{tab:part_2}
\small
\begin{tabular}{l *{2}{c}}
\toprule
$x$ & 6 & 7 \\
$m$ & 2 & 3 \\
$J$ & $\{1\}$ & $\{1,2\}$ \\
$\varphi_{J,m}$ & $-\phi_6+\phi_2+\phi_4$ & $-\phi_7+\phi_5+\phi_3-\phi_1+\phi_6-\phi_2-\phi_4$ \\
\midrule
$\ket{000}$ & 0 & 0 \\
$\ket{001}$ & 0 & 0 \\
$\ket{010}$ & 0 & 0 \\
$\ket{011}$ & 0 & 0 \\
$\ket{100}$ & 0 & 0 \\
$\ket{101}$ & 0 & 0 \\
$\ket{110}$ & 0 & 0 \\
$\ket{111}$ & $\phi_7-\phi_5-\phi_3+\phi_1-\phi_6+\phi_2+\phi_4$ & 0 \\
\bottomrule
\end{tabular}
\end{table}

\begin{theorem}
    \label{th_nqubits_3_methods}  Given a general quantum system of $N$ qubits $\ket{\mathbf{Q}} =\ket{Q_1 \dots Q_N}$, there exists a unitary transformation $\mathbf{U}$
    \begin{equation}
    \mathbf U= \left(\prod^1_{j=N} \prod^0_{k=j-1}U^z_{k,j}\right)\left(\prod^1_{j=N}\prod^0_{k=j-1}U^y_{k,j}\right)
    \end{equation}    
    such that  we can represent $\ket{\mathbf{Q}}$ with the following expression:
\begin{equation}
    \ket{\mathbf{Q}}=\mathbf{U} \ket{0}^{\otimes N}
\end{equation}

\end{theorem}

\begin{proof}
    
Since matrices $U^a_{k,j}$ are unitary it is sufficient to prove that  given an arbitrary state $\ket{\mathbf{Q}}$ there exist unitary $\mathbf{T}$ 
\begin{equation} \label{eq_unitary_inverse}
    \mathbf{T} =  \left(\prod^N_{j=1}\prod^{j-1}_{k=0}U^y_{k,j}\right) \left(\prod^N_{j=1}\prod^{j-1}_{k=0}U^z_{k,j}\right)
\end{equation}
such that
\begin{equation} \label{eq:inverse}
    \ket{0}^{\otimes N}=   \mathbf{T} \ket{\mathbf{Q}}.
\end{equation}
Once we do this, clearly $\mathbf{U}=\mathbf{T}^{\dag}$. Let $M = 2^N-1$ and 
$$\ket{\mathbf{Q}} =\sum_{x=0}^M A_xe^{i \phi_x} \ket{x}.$$\\
By Lemma \ref{le:less_thanZ}, there exists a transformation $\mathbf U_z$ such that in the system $\ket{\mathbf Q_1}=\mathbf U_z \ket{\mathbf{Q}} = \sum_{x=0}^M R_xe^{i\psi_x}\ket{x}$ we have that  $\psi_x = 0$ for all $0 \leq x \leq M$. Therefore it remains to prove that, given a state $\ket{\mathbf{Q_1}}$ where all the phases are zero, there exists a transformation 
$$\mathbf U_y = \prod^N_{j=1}\prod^{0}_{k=j-1}U^y_{k,j}$$ 
such that 
\begin{equation} \label{eq:inverse_to_zero}
\ket{0}^{\otimes N}=\mathbf U_y\ket{\mathbf{Q_1}}.
\end{equation}
Thus proving the existence of $\mathbf{T}=\mathbf U_y\mathbf U_z$  which satisfies \eqref{eq:inverse}. We prove this statement by induction on the number $N$ of qubits.

\vspace{1em}
\textit{Base case}: $N=1$. Clearly, if $\ket{ \mathbf Q_1} =\cos \frac{\theta_1}{2} \ket{0}+ \sin \frac{\theta_1}{2}\ket{1}$ then $\mathbf U_y = \left (R_y(\frac{\theta_1}{2} )\right)^{\dag} $ is the unitary transformation that satisfy \eqref{eq:inverse_to_zero}.

\vspace{1em}
\textit{Induction step}: $N>1$. Suppose the statement is true for all states of $N-1$ qubits having their phases equal to zero.  We prove that that there exists unitary transformations 
$\mathbf W_y =\prod^{N-1}_{k=0} U^y_{k,N} $  such that 
\begin{equation} \label{eq:inverseN}
    \ket{Q_1 \dots Q_{N-1}}\ket{ 0}=\mathbf W_y   \ket{Q_1 \dots Q_N}
\end{equation}
Once we do this, by induction hypothesis, there exists unitary transformation $\mathbf U_{y_{N-1}}$ such that $\ket{0}^{\otimes N-1} =\mathbf U_{y_{N-1}} \ket{Q_1 \dots Q_{N-1}}$, so that 

\begin{equation}
 \mathbf U_y=(\mathbf U_{y_{N-1}}\otimes I) \mathbf W_y
\end{equation} 
satisfies equation \eqref{eq:inverse_to_zero}.
The goal of transformation $\mathbf W_y$  is to make the amplitude of every odd component of $\ket{\mathbf Q_1}$ equal to zero.\\

Consider the procedure of Algorithm  \ref{alg:state_transformationY}.
\begin{figure}[htbp]
\caption{State Vector Transformation Procedure $Y$}
\label{alg:state_transformationY}
\begin{algorithmic}[1]
\Require The state vector $\ket{\mathbf Q_1} =\sum_{x=0}^M A_xe^{i \phi_x}\ket{x}$ where $\phi_x=0$ for all $0 \leq x \leq M$ 
\Ensure The set of angles of rotation $\theta_{J,N}$, $J \subseteq \{1,2, \dots, N-1\}$

\If { $A_{0} \neq 0$}
            \State Set $\theta_{\emptyset,N} \leftarrow 2\arctan \left (\frac{A_{1}}{A_{0}}\right )$  
        \Else 
            \State Set $\theta_{\emptyset,N} \leftarrow \pi$
        \EndIf
   
    \State Set $\ket{\mathbf Q_1} \leftarrow \bigwedge^{\emptyset}_N \left[R_y(\theta_{\emptyset,N}/2)\right]^{\dag} \cdot \ket{\mathbf Q_1}$      \For{$k = 1 \to N-1$}
    \For{all $J \subseteq \{1,\dots, N-1\}$ such that $|J|=k$}
    
        \State Let $x$ an $(N-1)$-dimensional binary vector such that $S(x)=J$
       
        \If { $A_{2x} \neq 0$}
            \State Set $\theta_{J,N} \leftarrow \arctan \left (\frac{A_{2x+1}}{A_{2x}}\right )$  
        \Else 
            \State Set $\theta_{J,N} \leftarrow \frac{\pi}{2}$
        \EndIf
        \State Set $\ket{\mathbf Q_1} \leftarrow \bigwedge^J_N \left[R_y(\theta_{J,N})\right]^{\dag} \cdot \ket{\mathbf Q_1}$  
        \label{step:transformY}
    \EndFor    
\EndFor
\State \Return $\theta_{J,N}$, $J \subseteq \{1,2, \dots, N-1\}$

\end{algorithmic}
\end{figure}
Note that at step \ref{step:transformY} of Algorithm \ref{alg:state_transformationY}, the amplitude of  of component with index $2x+1$ will be zero. In fact $\left[R_y(\theta_{J,N})\right]^{\dag}$ will act on component with index $2x+1$ by setting its amplitude as $$ -A_{2x}\sin(\theta_{J,N})+A_{2x+1}\cos(\theta_{J,N}).$$ Thus  if $A_{2x} \neq 0 $ and  $\theta_{J,N}=\arctan \frac{A_{2x+1}}{A_{2x}}$ then 
\begin{align*}
&-A_{2x}\sin(\theta_{J,N})+A_{2x+1}\cos(\theta_{J,N}) =     \\
& = -A_{2x} \tan(\theta_{J,N})\cos(\theta_{J,N}) +A_{2x+1}\cos(\theta_{J,N})\\
& = (-A_{2x} \tan(\theta_{J,N}) + A_{2x+1}) \cos(\theta_{J,N}) = 0.
\end{align*}
Otherwise, if $A_{2x} = 0$ setting $\theta_{J,N}=\pi/2$ also makes the amplitude of component $2x+1$ equal to zero. 
Let $J'$ be a set selected by the algorithm after $J$. We have that $|J'| \geq |J|$ and, therefore $J' \nsubseteq J$.  By Lemma \ref{le:subset}, the component  with index $2x+1$ will remain the same during the rest of the execution of the algorithm. 
Therefore at the end of the algorithm all the odd components of $\ket{\mathbf Q_1}$ have their amplitude equal to zero. The set of angles given in output therefore satisfy the requirement for the design of $\mathbf W_y$.
\end{proof}

\bigskip

\noindent
In the following, for the sake of simplicity we will use the rotation  matrix $R_y(\alpha)$:
\begin{equation*}
    R_y(\alpha) = \begin{pmatrix} 
        \cos(\alpha) & \sin(\alpha) \\ -\sin(\alpha) & \cos(\alpha) \end{pmatrix}.
\end{equation*}
As an example consider a system of $2$ qubits $$\ket{\mathbf Q_1}= A_0\ket{00}+A_1\ket{01}+A_2\ket{10}+A_3\ket{11},$$ where $A_j$ are positive real numbers.
At step 1 of the Algorithm \ref{alg:state_transformationY} we set $\theta_{\emptyset,2}=2\arctan\left(A_1/A_0\right)$. Then we apply the rotation $\bigwedge^{\emptyset}_2 R_y\left(\theta_{\emptyset,2}/2\right)$ to $\ket{\mathbf Q_1}$ whose matrix is

\begin{equation*}
\bigwedge^{\emptyset}_2 R_y\left(\frac{\theta_{\emptyset,2}}{2}\right)=
\begin{pmatrix}
\cos\left(\frac{\theta_{\emptyset,2}}{2}\right) & \sin\left(\frac{\theta_{\emptyset,2}}{2}\right) & 0 & 0 \\
-\sin\left(\frac{\theta_{\emptyset,2}}{2}\right) & \cos\left(\frac{\theta_{\emptyset,2}}{2}\right) & 0 & 0 \\
0 & 0 & \cos\left(\frac{\theta_{\emptyset,2}}{2}\right) & \sin\left(\frac{\theta_{\emptyset,2}}{2}\right) \\
0 & 0 & -\sin\left(\frac{\theta_{\emptyset,2}}{2}\right) & \cos\left(\frac{\theta_{\emptyset,2}}{2}\right)
\end{pmatrix}.
\end{equation*}
After this rotation, the components $A_j^{(1)}$ of the state vector $\ket{\mathbf Q_1}$ become
\begin{align*}
    A_0^{(1)} &= A_0\cos\left(\frac{\theta_{\emptyset,2}}{2}\right) + A_1\sin\left(\frac{\theta_{\emptyset,2}}{2}\right)\\
    A_1^{(1)} &=0\\
    A_2^{(1)} &=A_2\cos\left(\frac{\theta_{\emptyset,2}}{2}\right) + A_3\sin\left(\frac{\theta_{\emptyset,2}}{2}\right)\\
    A_3^{(1)} &=A_3\cos\left(\frac{\theta_{\emptyset,2}}{2}\right) -A_2\sin\left(\frac{\theta_{\emptyset,2}}{2}\right).
\end{align*}
Now at step 7 of Algorithm \ref{alg:state_transformationY} we set $\theta_{1,2}=\arctan \left(A_3^{(1)}/A_2^{(1)}\right)$ and we apply the rotation $\bigwedge^{1}_2 \left[R_y(\theta_{1,2})\right]^{\dag}$ to $\ket{\mathbf Q_1}$ whose matrix is
\begin{equation*}
\bigwedge^{1}_2 R_y(\theta_{1,2})=
\begin{pmatrix}
1 & 0 & 0 & 0 \\
0 & 1 & 0 & 0 \\
0 & 0 & \cos(\theta_{1,2}) & \sin(\theta_{1,2}) \\
0 & 0 & -\sin(\theta_{1,2}) & \cos(\theta_{1,2})
\end{pmatrix}.
\end{equation*}
After this rotation, the components $A_j^{(2)}$ of the state vector $\ket{\mathbf Q_1}$ become
\begin{align*}
    A_0^{(2)} &= A_0\cos\left(\frac{\theta_{\emptyset,2}}{2}\right) + A_1\sin\left(\frac{\theta_{\emptyset,2}}{2}\right)\\
    A_1^{(2)} &=0\\
    A_2^{(2)} &=A_2^{(1)}\cos\left(\theta_{1,2}\right) + A_3^{(1)}\sin\left(\theta_{1,2}\right)\\
    A_3^{(2)} &=0.
\end{align*}
Note that $A_0^{2}=A_0^{(1)}$ and $A_1^{2}=A_1^{(1)}$.\\
\noindent
As another example, consider a 3 qubits system. 
Then the rotation matrix $\bigwedge^{\emptyset}_3 R_y\left(\theta_{\emptyset,3}/2\right)$ is a $8 \times 8$ matrix (where the 0's are $2 \times 2$ zero matrices):
\begin{equation*}
\bigwedge^{\emptyset}_3 R_y\left(\frac{\theta_{\emptyset,3}}{2}\right) = \begin{pmatrix}
R_y\left(\frac{\theta_{\emptyset,3}}{2}\right) & 0 & 0 & 0 \\
0 & R_y\left(\frac{\theta_{\emptyset,3}}{2}\right) & 0 & 0 \\
0 & 0 & R_y\left(\frac{\theta_{\emptyset,3}}{2}\right) & 0 \\
0 & 0 & 0 & R_y\left(\frac{\theta_{\emptyset,3}}{2}\right)
\end{pmatrix}.
\end{equation*}
After the application of this rotation, the component with index $x=1$ of the system will be zero.
If $J=2$ then the rotation matrix $\bigwedge^{2}_3 R_y\left(\theta_{2,3}\right)$ is:
\begin{equation*}
\bigwedge^{2}_3 R_y\left(\theta_{2,3}\right)
 = \begin{pmatrix}
I & 0 & 0 & 0 \\
0 & R_y\left(\theta_{2,3}\right) & 0 & 0 \\
0 & 0 & I & 0 \\
0 & 0 & 0 & R_y\left(\theta_{2,3}\right)
\end{pmatrix}.
\end{equation*}
After the application of this rotation, the component with index $x=3$ of the system will be zero. If $J=1$ then the rotation matrix $\bigwedge^{1}_3 R_y\left(\theta_{1,3}\right)$ is:
\begin{equation*}
\bigwedge^{1}_3 R_y\left(\theta_{1,3}\right)
 = \begin{pmatrix}
I & 0 & 0 & 0 \\
0 & I  & 0 & 0 \\
0 & 0 & R_y\left(\theta_{1,3}\right) & 0 \\
0 & 0 & 0 & R_y\left(\theta_{1,3}\right)
\end{pmatrix}.
\end{equation*}
After the application of this rotation the component with index $x=5$ of the system will be zero. If $J = \{1,2\}$
then the rotation matrix $\bigwedge^{\{1,2\}}_3 R_y\left(\theta_{\{1,2\},3}\right)$ is:
\begin{equation*}
\bigwedge^{\{1,2\}}_3 R_y\left(\theta_{\{1,2\},3}\right)
 = \begin{pmatrix}
I & 0 & 0 & 0 \\
0 & I  & 0 & 0 \\
0 & 0 & I & 0 \\
0 & 0 & 0 & R_y\left(\theta_{\{1,2\},3}\right)
\end{pmatrix}.
\end{equation*}
After the application of this rotation the component with index $x=7$ of the system will be zero.

\bibliographystyle{quantum}
\bibliography{bibliografia}

\end{document}